\documentclass[conference]{IEEEtran}
\usepackage{citesort}
\usepackage[dvips]{graphicx}
\usepackage[cmex10]{amsmath}
\usepackage{algorithm}
\usepackage{algpseudocode}
\usepackage{array}
\usepackage[tight,footnotesize]{subfigure}
\usepackage{setspace}
\usepackage{enumerate} 
\usepackage{url}

\usepackage{amsfonts, amssymb, amsthm, mathrsfs}
\usepackage{setspace}
\usepackage{Tabbing}
\usepackage{fancyhdr}
\usepackage{lastpage}
\usepackage{extramarks}
\usepackage{chngpage}
\usepackage{soul,color}
\usepackage{indentfirst}
\usepackage{float,wrapfig}
\usepackage{times}
\usepackage{multirow}
\usepackage{bm}
\usepackage{bbm}

\newtheorem{defi}{Definition}
\newtheorem{assp}{Assumption}
\newtheorem{prop}{Proposition}
\newtheorem{remark}{Remark}

\newcolumntype{I}{!{\vrule width 1.5pt}}
\newlength\savedwidth

\newcommand{\sinr}{\textrm{SINR}}

\begin{document}

\title{On/Off Macrocells and Load Balancing in Heterogeneous Cellular Networks}

\author{\IEEEauthorblockN{Qiaoyang Ye$^*$, Mazin Al-Shalash$^\dagger$, Constantine Caramanis$^*$, and Jeffrey G. Andrews$^*$}

\IEEEauthorblockA{$\ ^*$Dept. of ECE, The University of Texas at Austin\\
$\ ^\dagger$Huawei Technologies\\
Email: qye@utexas.edu, mshalash@huawei.com, constantine@utexas.edu,  jandrews@ece.utexas.edu}}

\maketitle

\begin{abstract}

The rate distribution in heterogeneous networks (HetNets) greatly benefits from load balancing, by which mobile users are pushed onto lightly-loaded small cells despite the resulting loss in SINR. This offloading can be made more aggressive and robust if the macrocells leave a fraction of time/frequency resource blank, which reduces the interference to the offloaded users.  We investigate the joint optimization of this technique -- referred to in 3GPP as enhanced intercell interference coordination (eICIC) via almost blank subframes (ABSs) -- with offloading in this paper.  Although the joint cell association and blank resource (BR) problem is nominally combinatorial, by allowing users to associate with multiple base stations (BSs), the problem becomes convex, and upper bounds the performance versus a binary association.  We show both theoretically and through simulation that the optimal solution of the relaxed problem still results in an association that is mostly binary. The optimal association differs significantly when the macrocell is on or off; in particular the offloading can be much more aggressive when the resource is left blank by macro BSs.  Further, we observe that jointly optimizing the offloading with BR is important.  The rate gain for cell edge users (the worst 3-10\%) is very large -- on the order of 5-10x -- versus a naive association strategy without macrocell blanking.
\end{abstract}

\IEEEpeerreviewmaketitle

\section{Introduction}
Future wireless networks are evolving to become heterogeneous, with the proliferation of  small base stations (BSs) such as picocells and femtocells which differ significantly in terms of transmit power, coverage area, physical size, and other capabilities \cite{GhoAnd12}.  In particular, the massive differences in the coverage areas of different cells makes load balancing very important.  Without load balancing -- that is, where users simply associate with the strongest downlink pilot signal -- macrocells remain the bottleneck and small cells are extremely underutilized \cite{Bje11,DamMon11}.  However, load balancing results in the offloaded users experiencing not only a weaker received signal, but also stronger interference.  

This motivates the straightforward idea of leaving certain time/frequency resource of the macrocells blank, during which the offloaded users can receive much higher SINR from the small cells.  Although this decreases the time-frequency resources available to the remaining macrocell users, if there is enough parallelism in the shorter range small cell transmissions (which now have higher SINR and thus rate), this loss can be overcome, and indeed the net gain can be fairly significant~\cite{LopGuv11}.  For example, in 3GPP, these ``off'' periods  are called \emph{almost blank subframes} (ABSs), during which only the reference signals are transmitted (to allow mobiles to still monitor and synchronize with the macro BS), while the rest of the resource blocks (RBs) in the frame are unused.  Such a scheme motivates a few fundamental questions.  What fraction of the resource should the BSs leave blank?  How should users associate during each of the ``on'' and ``off'' periods as a function of the small cell density and other system parameters?  What is the best-case gain of such an approach?

\textbf{Related Work.}
There has been some recent effort to study these questions. In our recent paper \cite{YeRon12}, an optimal user association  to achieve load balancing in HetNets is found, but without consideration of eICIC. A heuristic algorithm for cell association and resource partitioning was proposed in \cite{MadBor10}, which is seen to improve the performance in HetNets. In \cite{JoSan11}, a tractable framework  for SINR analysis in HetNets with range expansion is proposed. Paper \cite{GuvJeo11} studied the impact of range expansion and proposed a heuristic cell selection scheme, which provides a compromise between sum capacity and fairness. Another cell selection scheme with a given ABS ratio is proposed in \cite{OhHan12}, which is based on the SINR without consideration of the load of BSs. Nevertheless, how to jointly optimize the user association and allocation of blank resource (BR) or other eICIC approaches is still an open issue.

\textbf{Contributions.}
Finding a true performance bound for joint optimization of  load-aware user association and eICIC is a very challenging problem due to the coupled relationship between the users' association, scheduling, and eICIC. This paper extends the framework in~\cite{YeRon12} to the eICIC case, where all macro BSs send blank RBs synchronously and for the same fraction of resource. The joint optimization is combinatorial if users can only associate with one BS, but if this constraint is relaxed to allow users to associate for a fraction of time with different BSs, the resulting problem turns out to be convex.  It also upper bounds the network utility with the binary association. Further, we prove that the number of users associated with multiple BSs is quite limited, and is in fact smaller than the number of BSs. Therefore, a binary association would be expected to have comparable performance. 

We then turn our attention to the optimal user association during the two different phases.  We demonstrate that the optimal association for the On and Off periods is very different, with much more load balancing achieved during the Off periods. The fraction of blank (Off) time/frequency RBs is found to increase with the number of picocells in the network.  For example, if there are 6-10 picocells per macro, then the macrocell should be Off about half the time.  The gains from jointly optimizing the load balancing with the BR is quite large, while without an appropriately modified association, the gain from introducing BR is limited.
\section{System Models}\label{sec:model}

We consider a downlink HetNet with $K$ tiers of BSs, where each tier models a typical type of BSs. We consider a synchronous configuration, where each macro BS has the same blank RBs. We jointly optimize the duty cycle of BR and the corresponding user association.  The sets of all BSs and users are denoted by $\mathcal{B}$ and $\mathcal{U}$ with size $N_B$ and $N_U$, respectively. Let $\mathcal{B}_1\in\mathcal{B}$ be the set of marcocell BSs, with size $N_{B_1}$. The SINR of user $i$ from BS $j$ in normal (On) RBs is
\begin{equation}\label{eq:SINRon}
\text{SINR}_{ij}^{(n)}=\frac{P_j h_{ij}}{\sum_{n\in\mathcal{B}/j} P_n h_{in}+\sigma^2}, \quad \forall i\in\mathcal{U}, j\in\mathcal{B},
\end{equation}
while the SINR of user $i$ from BS $j$ in blank (Off) RBs is
\begin{equation}\label{eq:SINRoff}
\text{SINR}_{ij}^{(b)}=\left\{
\begin{aligned}
& \frac{P_j h_{ij}}{\sum_{n\in\mathcal{B}/(\mathcal{B}_1\cup j)} P_n h_{in}+\sigma^2}, \quad \forall j\in\mathcal{B}/\mathcal{B}_1,\\
& 0, \quad \forall i\in\mathcal{U}, j\in\mathcal{B}_1,
\end{aligned}
\right.
\end{equation}
where $P_j$ denotes the transmit power of BS $j$, $h_{ij}$ is the channel gain of the link from BS $j$ to user $i$, and $\sigma^2$ is the noise power level. The channel gain includes path loss, shadowing and antenna gain. In this paper, we assume a static channel during each resource allocation period, which is applicable for low mobility environments. Stochastic channel analysis is left as future work.


We denote by $c_{ij}^{(n)}$ and $c_{ij}^{(b)}$ the spectral efficiency of user $i$ from BS $j$ in normal  and blank  RBs, respectively.  Generally, spectral efficiency is a logarithmic function of SINR (e.g., $c_{ij}^{(n)}=\log\left(1+\sinr_{ij}^{(n)} \right)$). Denoting the fraction of resources allocated from BS $j$ to user $i$ in normal and blank RBs by $s^{(n)}_{ij}$ and $s_{ij}^{(b)}$, respectively, where  $\sum_{i\in\mathcal{U}} s_{ij}=1$, we can define the long-term rate as follows.

\begin{defi}\label{def:rate}
The long-term rate of user $i$ from BS $j$ is
\begin{equation}
R_{ij}=(1-z)s_{ij}^{(n)}c_{ij}^{(n)}+zs_{ij}^{(b)}c_{ij}^{(b)},
\end{equation}
where $z$ is the fraction of blank RBs. The overall rate of user $i$, denoted by $R_i$, can be calculated according to $R_i=\sum_{j\in\mathcal{B}} R_{ij}$.
\end{defi}

In the following section, we investigate a utility maximization problem in terms of the long-term rate $R_{i}$ to find the optimal blank RB ratio, and the corresponding optimal user association.

\section{Problem Formulation}\label{sec:formulation}

The resource allocation variables $s_{ij}^{(n)}$ and $s_{ij}^{(b)}$ also indicate the association (i.e., user $i$ is associated with BS $j$ in normal RBs when $s_{ij}^{(n)}>0$). Typically, each user will be served by at most one BS, i.e., $\sum_j \mathbf{1}_{\{ s_{ij}^{(n)} >0\}}\leq1$ and $\sum_j \mathbf{1}_{\{ s_{ij}^{(b)} >0\}}\leq 1$, termed  ``single association'' in this paper. The single association constraint makes the problem combinatorial, and thus difficult to solve. Though it may not be viable in practice to allow users to be served by multiple BSs at the same time, we relax the single association constraint and thus make the problem convex, which can serve as an upper bound to benchmark the performance. In the remaining of this paper, we make the following assumption.
\begin{assp}\label{as:more}
Users can be jointly served by more than one BS at the same time.
\end{assp}
We call users associated with multiple BSs ``fractional users'' and the relaxed association ``fractional user association''. Under the above assumption, the single association constraint is relaxed, and the resulting optimization problem is:
\begin{equation}\label{eq:generalutility}
\begin{aligned}
 \max\limits_{s^{(n)},s^{(b)},z} & \quad \sum\limits_{i\in \mathcal{U}} U_{i}(R_i)\\
 \textrm{s.t.} & \quad \sum_i s_{ij}^{(n)}\leq 1, \ \forall j,\\
 & \quad \sum_i s_{ij}^{(b)}\leq 1, \ \forall j,\\
 &\quad s_{ij}^{(n)}, s_{ij}^{(b)}\in [0,1] ,\  \forall i,j\\
 &\quad z\in [0,1], \ \forall j,
\end{aligned}
\end{equation}
where $U_i(\cdot)$ is a continuously differentiable, and strictly concave utility function~\cite{StaWic09}. We adopt a logarithmic utility function, which naturally achieves load balancing and can be viewed as a sort of proportional fairness~\cite{YeRon12}. Changing $x_{ij}=zs_{ij}^{(n)}$ and $y_{ij}=(1-z)s_{ij}^{(b)}$, the optimization problem~(\ref{eq:generalutility}) is equivalent to
\begin{equation}\label{eq:abscvx}
\begin{aligned}
 \max\limits_{x,y,z} & \quad \sum\limits_{i\in \mathcal{U}} \log\left(\sum_{j\in\mathcal{B}} \left(x_{ij}c_{ij}^{(n)}+y_{ij}c_{ij}^{(b)}\right)\right)\\
 \textrm{s.t.} & \quad \sum_{i\in \mathcal{U}} x_{ij}\leq 1-z, \ \forall j,\\
 & \quad \sum_{i\in \mathcal{U}} y_{ij}\leq z, \ \forall j,\\
 &\quad x_{ij}, y_{ij},z\in [0,1] ,\  \forall i,j.\\
\end{aligned}
\end{equation}

\begin{prop}\label{prop:ABScvx}
The optimization problem (\ref{eq:abscvx}) is convex.
\end{prop}

\begin{proof}
Denote the objective function in (\ref{eq:abscvx}) by $g(x,y)$. We will use Hessian matrix to check its convexity. The Hessian has the form
\begin{equation}\label{eq:hessmatrix}
\nabla^2 g=-
\begin{bmatrix}
 B_{1} & 0  & \cdots & 0 \\
 0 & B_{2}  & \cdots & 0 \\
 \vdots  &\vdots & \ddots   &\vdots \\
 0 & 0 & 0  & B_{N_U}\\
 \end{bmatrix}.
 \end{equation}

The matrix $B_i$ can be expressed as
\begin{equation}
 B_i=\frac{\mathbf{c}^T\mathbf{c}}{\left(\sum_k \left(x_{ik}c_{1_{ik}}+y_{ik}c_{2_{ik}}\right)\right)^2},
\end{equation} 
where $\mathbf{c}=\left[c_{1_{i1}}, c_{1_{i2}}, \cdots,  c_{1_{iN_B}}, c_{2_{i1}}, c_{2_{i2}}, \cdots,  c_{2_{iN_B}}  \right]$.

Therefore, the matrix $B_i$ is positive semi-definite (PSD) for all $i$, and thus $-\nabla^2 g$ is also PSD. The problem (\ref{eq:abscvx}) has a concave objective function with linear constraints, which implies that (\ref{eq:abscvx}) is a convex optimization.
\end{proof}

Though the objective function in (\ref{eq:abscvx}) is \emph{strictly} concave with respect to $R_i$, it is not strictly concave with respect to $x$ and $y$. We have the following proposition which shows the uniqueness of optimal solution.
\begin{prop}\label{prop:uniquesol}
The optimization problem (\ref{eq:abscvx}) has a unique optimal resource allocation (i.e., unique $x_{ij}^*$ and $y_{ij}^*$) almost surely. If $\sum_i \left(x_{ij}^*+y_{ij}^*\right)=1$, then the optimal solution of problem (\ref{eq:abscvx}) is unique (i.e., $z^*$ is also unique).
\end{prop}
\begin{proof}\label{pf:prop-uniquesol}
The proof includes two basic steps. In the first step, $U(R_i)$ is strictly concave in $R_i$ and thus we have unique solution $R_i^*$. The users with single association can be obtained uniquely. The second step is to show that the associations of fractional users can also be uniquely generated from $R_i^*$. This can be proved by bipartite graphs. Details can be found in \cite{GajHua09}. When $\sum_i \left(x_{ij}^*+y_{ij}^*\right)=1$, we have $\sum_{i\in \mathcal{U}} x^*_{ij}=1-z^*$ and thus $z^*$ is also unique.
\end{proof}

Returning to Assumption \ref{as:more}, what is the impact of the relaxation on the optimal solution? To answer this question, we first use a graph to represent the association, and then by applying  \textit{Karush-Kuhn-Tucker} (KKT) conditions, we show that the impact is limited.

In the graph representation of association, the nodes correspond to the users in HetNets, while the edges correspond to the BSs shared between the connected users, illustrated in Fig. \ref{fig:association}. Each node has a unique ID from $1$ to $N_U$, which is the indicator of users, and each edge has a color from $1$ to $N_B$ for BS identification. For example, in Fig. \ref{fig:associate3}, user $i$ is associated with both BS $j$ and $n$, and user $m$ is jointly served by BS $j$ and $k$. Note that the graph is not necessarily connected. The number of isolated subfigures depends on the number of fractional users. Another important property of the representation graph is that it is comprised of several connected/isolated \textit{complete graphs}.

\begin{figure}
\centering
\setcounter{subfigure}{0}
\subfigure[Example of three users.]{\label{fig:associate3}\includegraphics[width=4cm, height=4cm]{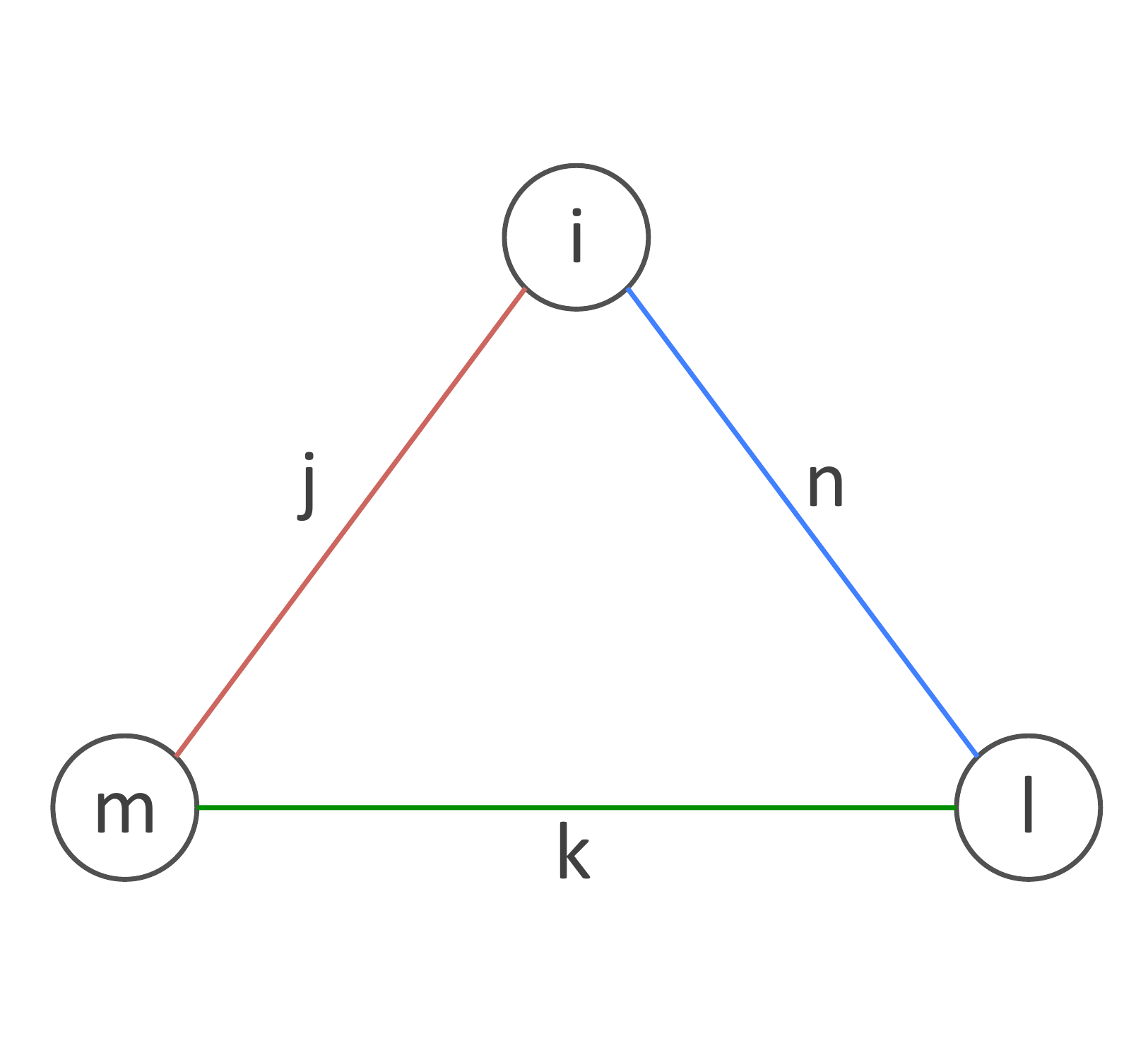}}
 \hspace{0.1in}
\subfigure[Example of four users.]{\label{fig:associate4}\includegraphics[width=4cm, height=4cm]{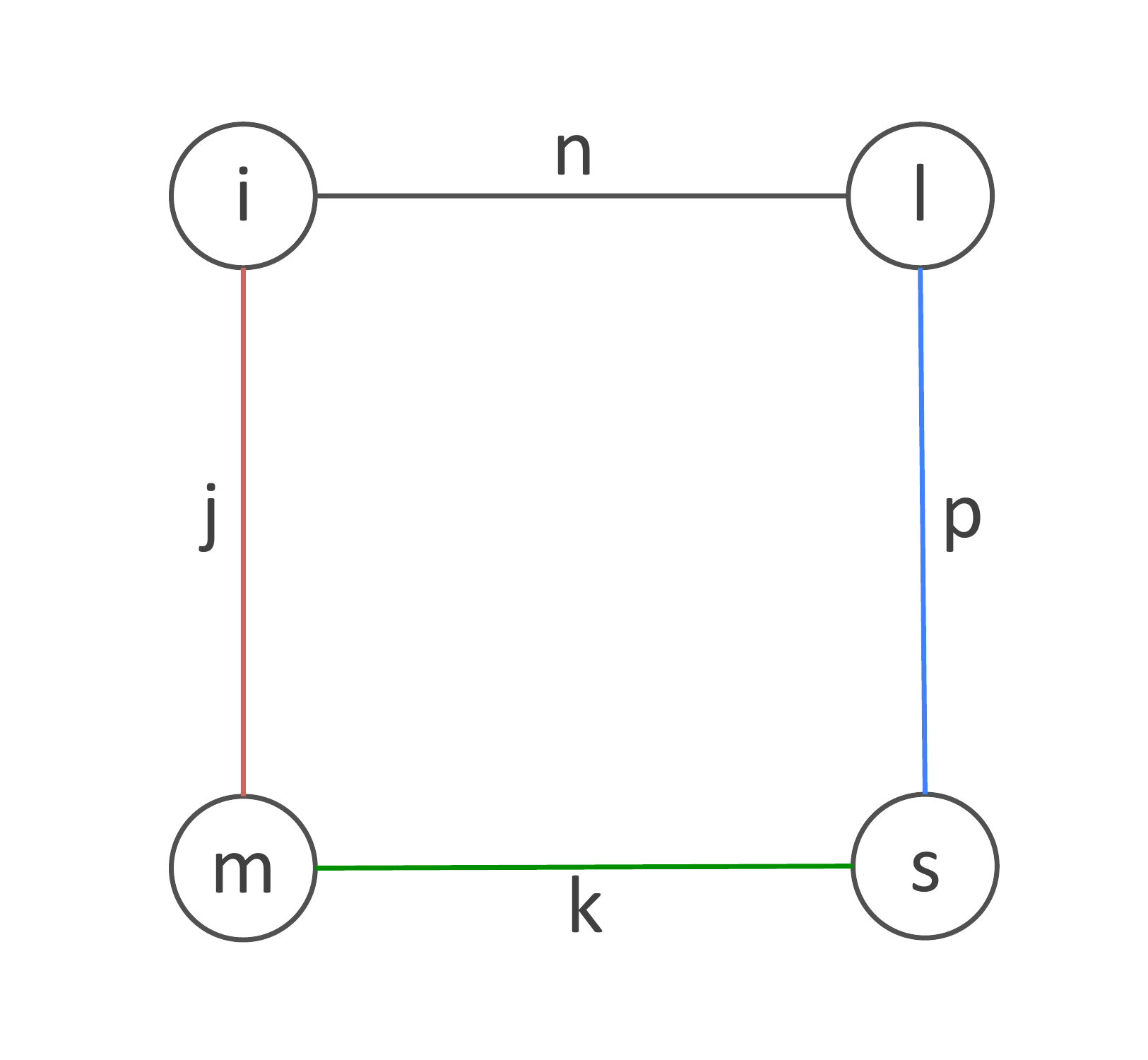}}
\caption{Examples of graph representation which illustrates the fractional user association.}
\label{fig:association}
\end{figure}

 
The convex optimization (\ref{eq:abscvx}) has differentiable objective and constraint functions, and satisfies Slater's condition. Therefore, the KKT conditions provide necessary and sufficient conditions for the optimality  \cite{BoyVan04}. Applying  the KKT conditions to problem (\ref{eq:abscvx}), we have the following proposition.

\begin{prop}\label{prop:kkt}
In the optimal solution, the number of users which are served by multiple BSs in normal RBs is at most~$N_B-1$. In blank RBs, the number of users associated with multiple BSs is at most~$N_B-N_{B_1}-1$.
\end{prop}
\begin{proof}
We adopt similar techniques in \cite{GajHua09}. For completeness, we provide the proof as follows. We define the \textit{Lagrangian} associated with problem~(\ref{eq:abscvx})  as
\begin{equation}
\begin{aligned}
&L(x,y, z, \lambda,\nu)= -\sum\limits_{i\in \mathcal{U}} \log\left(\sum_{j\in\mathcal{B}} \left(x_{ij}c_{ij}^{(n)}+y_{ij}c_{ij}^{(b)}\right)\right) \\
&+\sum_{j\in\mathcal{B}} \lambda_j\left(\sum_{i\in \mathcal{U}} x_{ij}-z \right)+\sum_{j\in\mathcal{B}} \nu_j\left(\sum_{i\in \mathcal{U}} y_{ij}-(1-z) \right),
\end{aligned}
\end{equation}
where $\lambda_j$ and $\nu_j$ are the \textit{Lagrange multipliers} associated with the $j$th inequality constraint in normal and blank RBs in~(\ref{eq:abscvx}), respectively. The KKT conditions are:
\begin{equation}\label{eq:kkt}
\left\{
\begin{aligned}
&\frac{c_{ij}^{(n)}}{R_i}=\lambda_j, \text{ if } x_{ij} > 0,\\
&\frac{c_{ij}^{(b)}}{R_i}=\nu_j, \text{ if } y_{ij} > 0, \\
&{\sum_j \lambda_j}={\sum_j \nu_j}, \text{ if } z \in (0, 1),\\
&\ \sum_i x_{ij}\leq z, \  \lambda_j\left(\sum_i x_{ij}-z \right)=0,\\
&\sum_i y_{ij}\leq 1-z, \ \nu_j\left(\sum_i y_{ij}-(1-z) \right)=0,\\
& x_{ij}, y_{ij}, z\in [0,1], \lambda_j,\nu_j \geq 0\\
\end{aligned}
\right.
\end{equation}

We conduct analysis on normal RBs, and the same conclusion can be extended to the blank RBs. From KKT conditions (\ref{eq:kkt}), for $x_{ij}>0, x_{in}>0, x_{mj}>0$ and $x_{mn}>0$, we have
\begin{equation}
\frac{c_{ij}^{(n)}}{c_{in}^{(n)}}=\frac{\lambda_j}{\lambda_n}=\frac{c_{mj}^{(n)}}{c_{mn}^{(n)}},
\end{equation}
which is true with probability 0. Therefore, it is almost sure that any two users can share at most one same BS (i.e., the number of edges between any two nodes in  graph is at most 1). Similarly, we consider an example of three users, illustrated in Fig. \ref{fig:associate3}. There are three possible cases:
\begin{enumerate}
\item BSs $j,n,k$ are three different BSs:
We have
\begin{equation}
\frac{c_{mj}^{(n)}}{c_{mk}^{(n)}}=\frac{\lambda_j}{\lambda_n}\frac{\lambda_n}{\lambda_k}=\frac{c_{ij}^{(n)}}{c_{in}^{(n)}}\frac{c_{ln}^{(n)}}{c_{lk}^{(n)}},
\end{equation}
which is true with probability 0.
\item $n=k\neq j$:

The user $m$ is associated with BS $j$ and $n$, and the user $i$ is also associated with $j$ and $n$, which contradicts the result in the two-user example.

\item $j=n=k$:
It is possible that these three users are all associated with the same BS, where the representation graph becomes a complete graph.
\end{enumerate}

Therefore, a graph representation of three users contains either a loop with the same color or no loop. We can get a similar result for a graph with more than three users (e.g., Fig.~\ref{fig:associate4}). In conclusion, the users associated with the same BS constitute a complete graph with edges on the same color. We can generate a new graph, where each complete graph can be considered as a new node. The new graph has no loops and thus it has the maximal number of edges when it is a tree. The number of edges in a tree is one less than the number of nodes in the tree. Therefore, the maximal number of edges connecting different complete graphs is $N_B-1$. The number of users associated with more than one BSs equals the number of edges in the new graph, which is no more than $N_B-1$. We can get similar conclusions for the blank RBs.
\end{proof}

Although we relax the single association constraint, Proposition \ref{prop:kkt} indicates that the relaxed solution would be close to a binary association in practice. This implies the possibility to get a well-approximated near-optimal single association solution via rounding. From KKT conditions, we can also conclude Proposition~\ref{prop:sameass}, about the difference between resource allocations in normal and blank RBs.

\begin{prop}\label{prop:sameass}
The number of users which get resources from the same BS in normal and blank RBs is at most ~$N_B-N_{B_1}$.
\end{prop}
\begin{proof}\label{pf:prop-sameass}
According to KKT conditions (\ref{eq:kkt}), if user $i$ and $m$ are associated with BS $j$ at both normal and blank RBs, we have
\begin{equation}
\frac{c_{ij}^{(n)}}{c_{ij}^{(b)}}=\frac{\lambda_j}{\nu_j}=\frac{c_{mj}^{(n)}}{c_{mj}^{(b)}},
\end{equation}
which is true with probability 0. Therefore, it is almost surely true that no more than two users can connect to a BS both in normal and blank RBs. In other words, each BS serves at most one user in both normal and blank RBs.
\end{proof}

\begin{remark}
Proposition \ref{prop:sameass} implies that the resource allocation in normal RBs is very different from the blank RBs. Only a small fraction of users keep the same association in both normal  and blank RBs.
\end{remark}

\section{Performance Evaluation}\label{sec:simulation}

\begin{table}
\caption{Simulation parameters}\label{tb:simu}
\begin{center}
\begin{tabular}{|l||c|}
\hline
Macrocell layout & Hexagonal grid\\
Pico/femtocell/UE  distribution & PPP\\
Density of macros & $1/500^2$ m$^{-2}$\\
Density of picos & $4/500^2$ m$^{-2}$\\
Density of femtos & $12/500^2$ m$^{-2}$\\
Density of cellular users & $80/500^2$ m$^{-2}$\\
\hline\hline
Transmit power of macros & $40$ W\\
Transmit power of picos & $1$ W\\
Transmit power of femtos & $0.1$ W\\
\hline
Noise  power & $-124$ dBm\\
Path loss exponent & $3.5$\\
Fading & Rayleigh \\
\hline
\end{tabular}
\end{center}
\end{table}

In this section, we provide simulation results to validate the proposed framework and analytical results. The main simulation parameters used are summarized in Table \ref{tb:simu} unless otherwise specified.

\begin{figure}
\centering
\setcounter{subfigure}{0}
\subfigure[Max-SINR association]{\label{fig:maxsinr}\includegraphics[width=8.5cm, height=5.2cm]{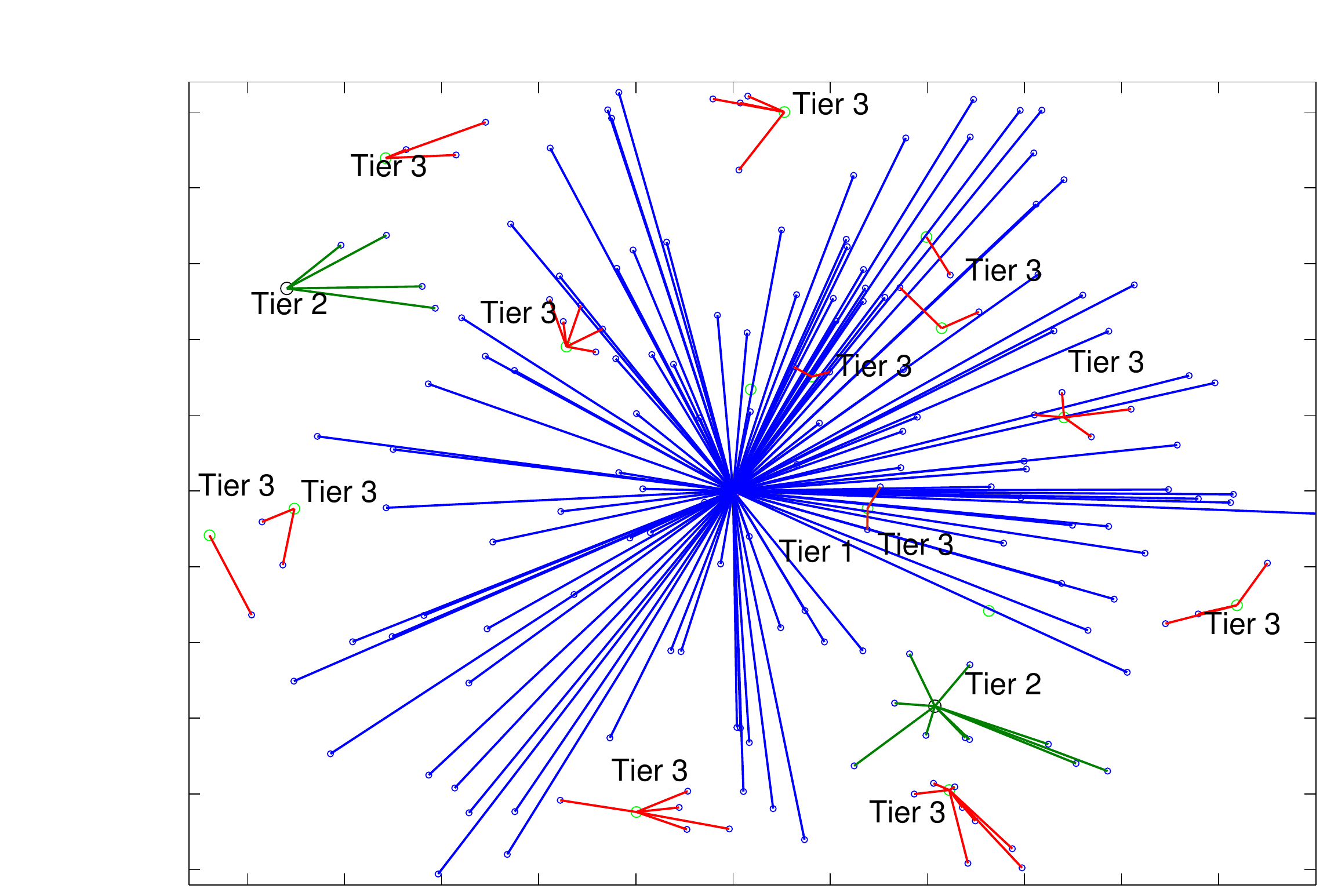}}
 \hspace{0.1in}
\subfigure[Load-aware association without blank resource.]{\label{fig:lbeg}\includegraphics[width=8.5cm, height=5.2cm]{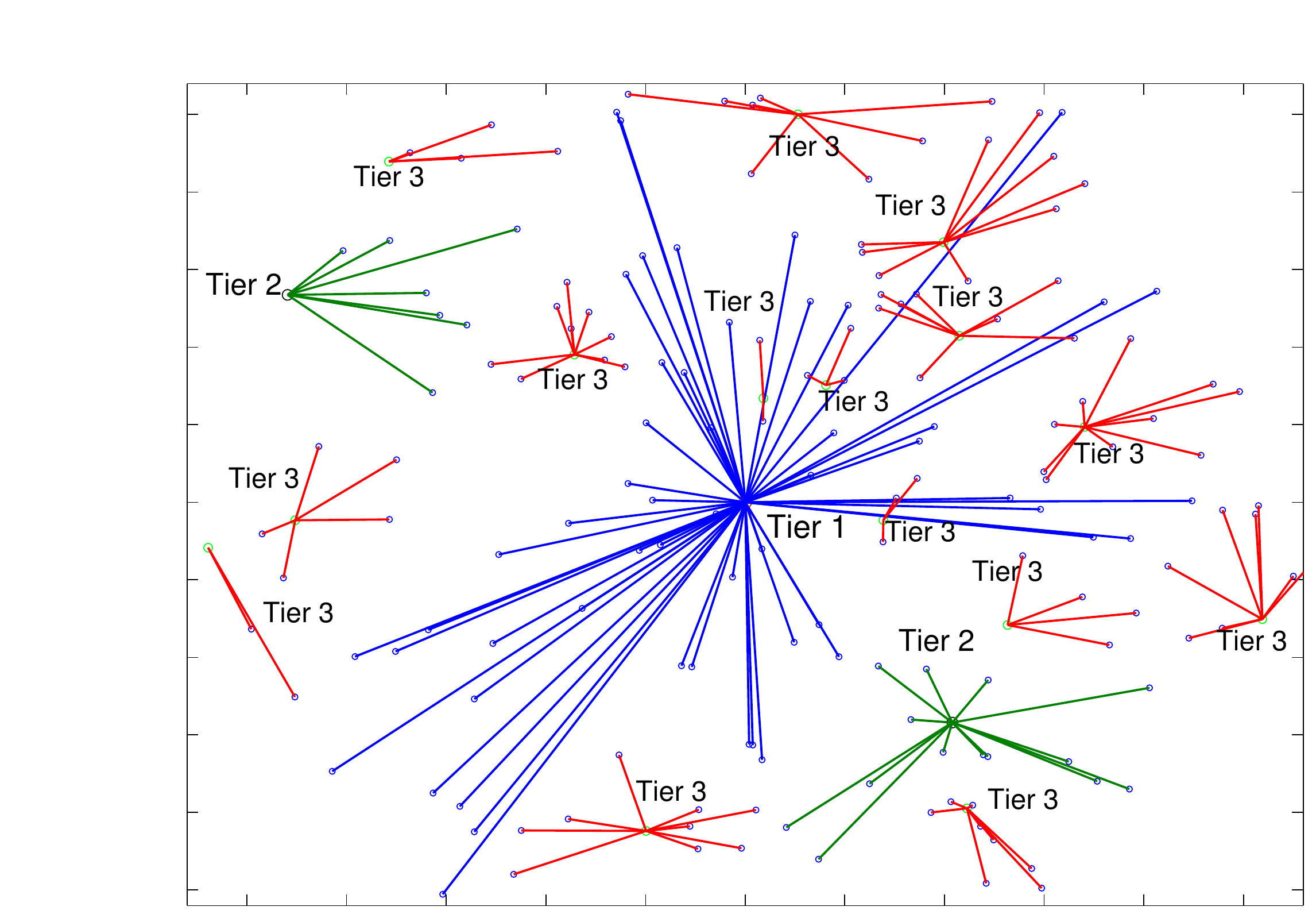}}
\subfigure[Optimal association with blank resource.]{\label{fig:abseg}\includegraphics[width=8.5cm, height=5.2cm]{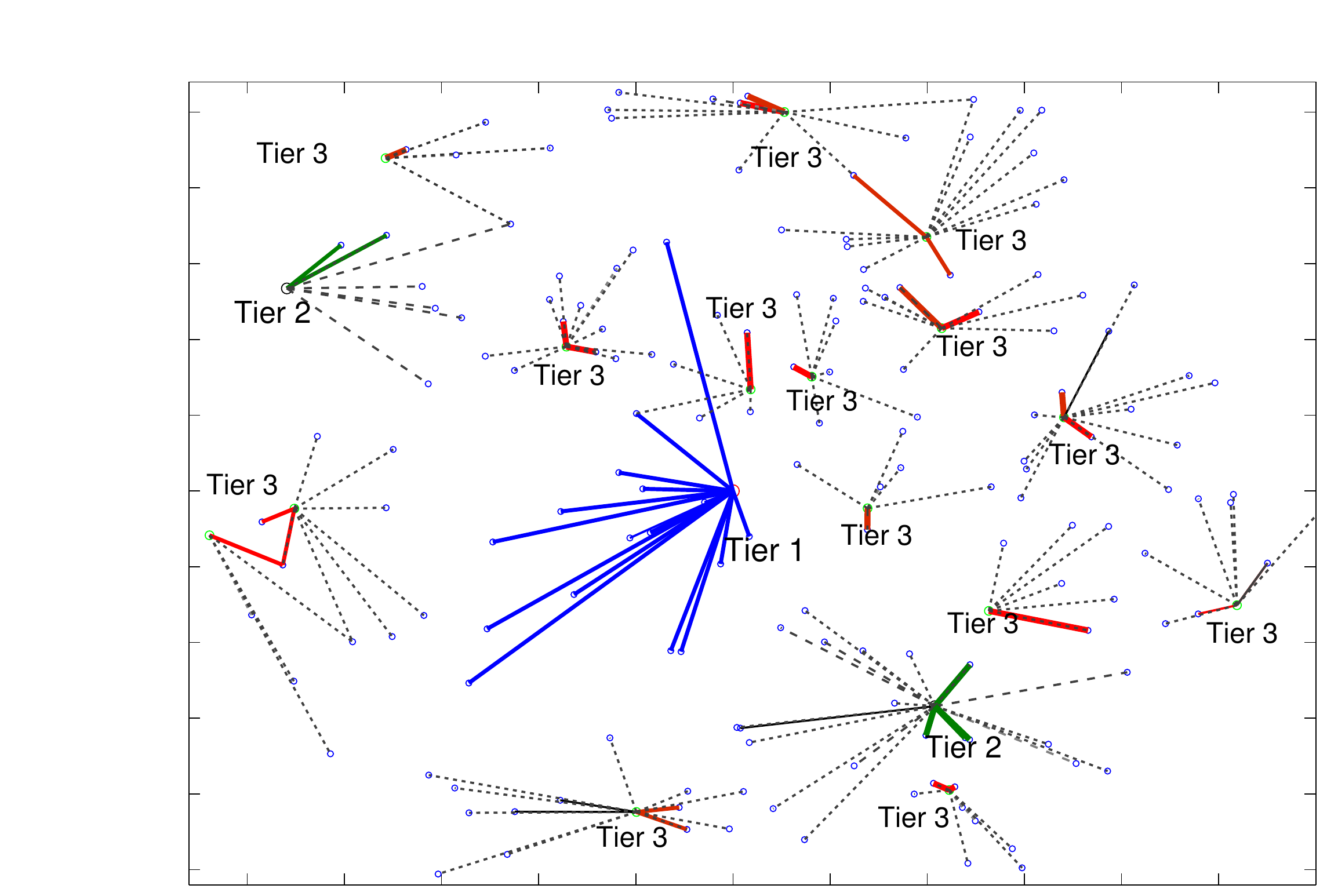}}
 \hspace{0.1in}
\caption{Examples of associations in HetNets with different association schemes. The dashed lines indicate the association in blank RBs, while the solid lines indicate the association in normal RBs. }
\label{fig:associationeg}
\end{figure}

Fig. \ref{fig:associationeg} shows examples of associations with different schemes. In the conventional user association scheme shown in Fig. \ref{fig:maxsinr}, the load is very unbalanced. Most of the users are associated with the macrocell, but may get small rates even with strong SINR. The load-aware user association in Fig.~\ref{fig:lbeg} achieves  more balanced load, and thus leads to a more efficient resource utilization. Adopting the resource blanking, users can be served in either the normal or/and blank RBs. The associations in blank and normal RBs are illustrated by dashed lines and solid lines in Fig. \ref{fig:abseg}, respectively. We can verify our propositions that the number of fractional users is very small, and the associations in blank and normal RBs are very different. More users would be served by small cells in blank RBs, where the strong macrocell interference can be avoided.

\begin{figure}
\centering
\includegraphics[width=8cm, height=6cm]{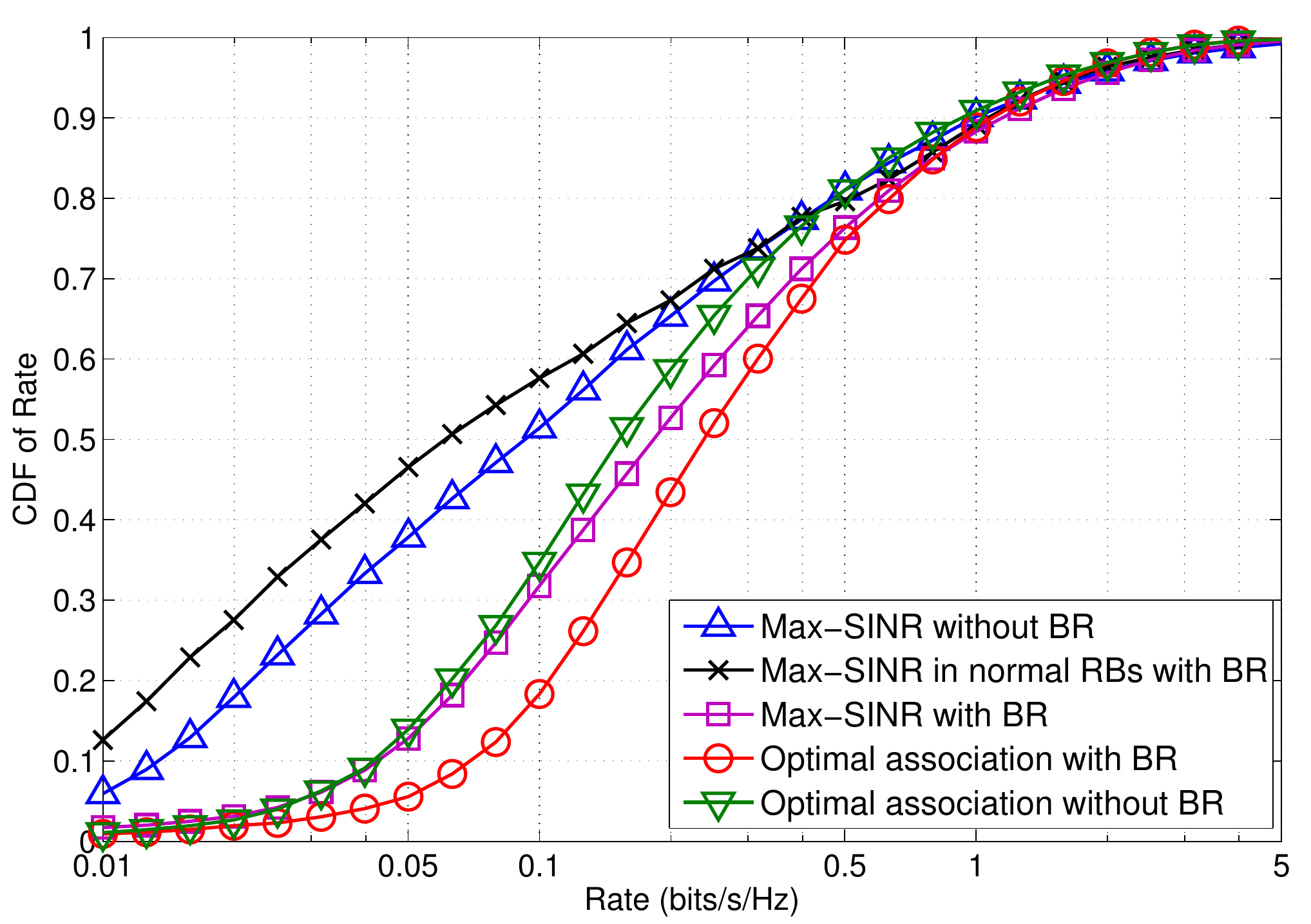}
\caption{The rate distribution of users with different association schemes.}
\label{fig:sinr}
\end{figure}

\begin{figure}
\centering
\setcounter{subfigure}{0}
\subfigure[Load of optimal association with resource blanking in blank RBs vs. load of optimal assocation without resource blanking.]{\label{fig:loadabs}\includegraphics[width=8.4cm, height=6cm]{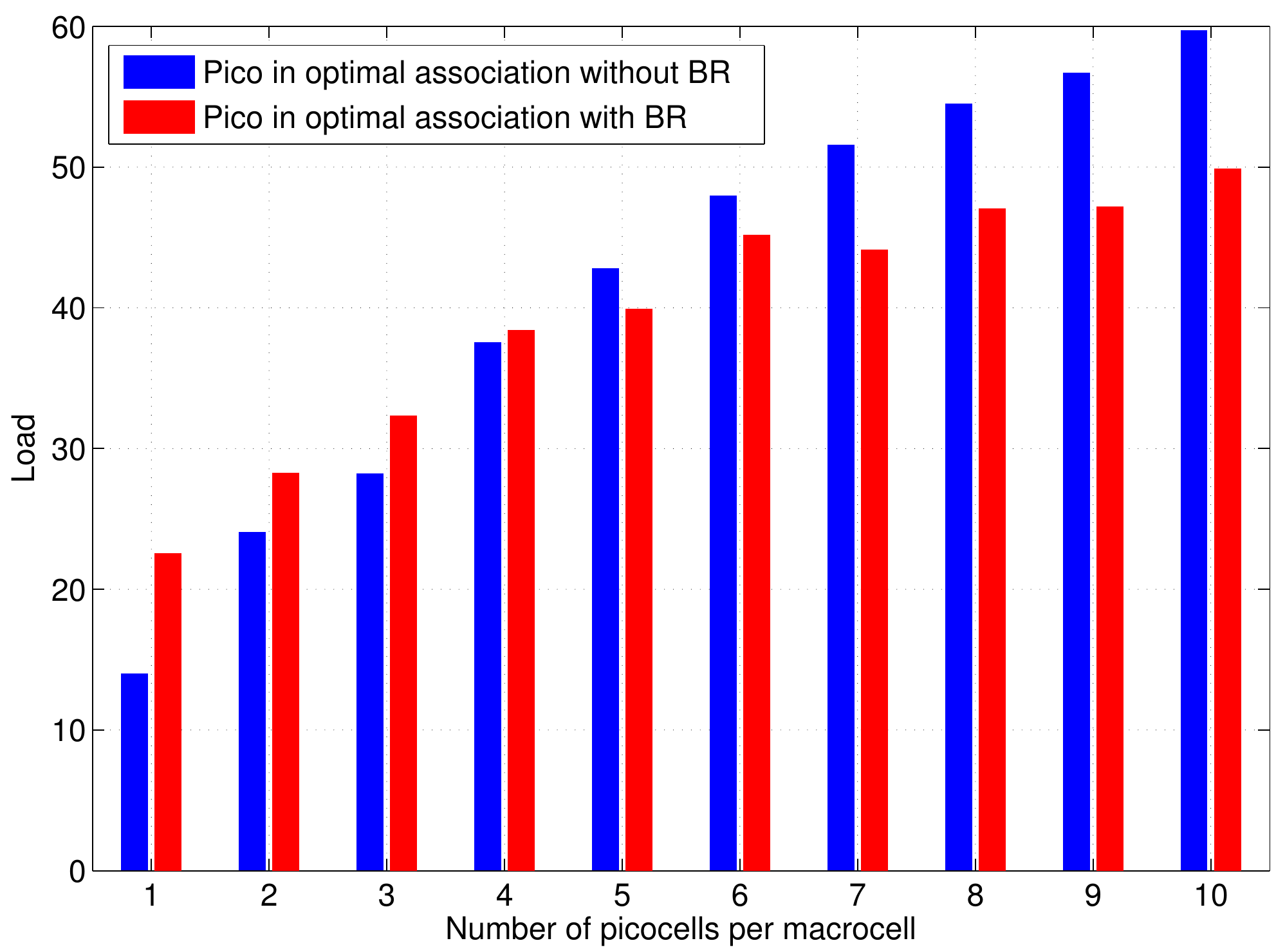}}
 \hspace{0.1in}
\subfigure[Load of optimal association with resource blanking in noraml RBs vs. load of optimal assocation without resource blanking.]{\label{fig:loadnormal}\includegraphics[width=8.4cm, height=6cm]{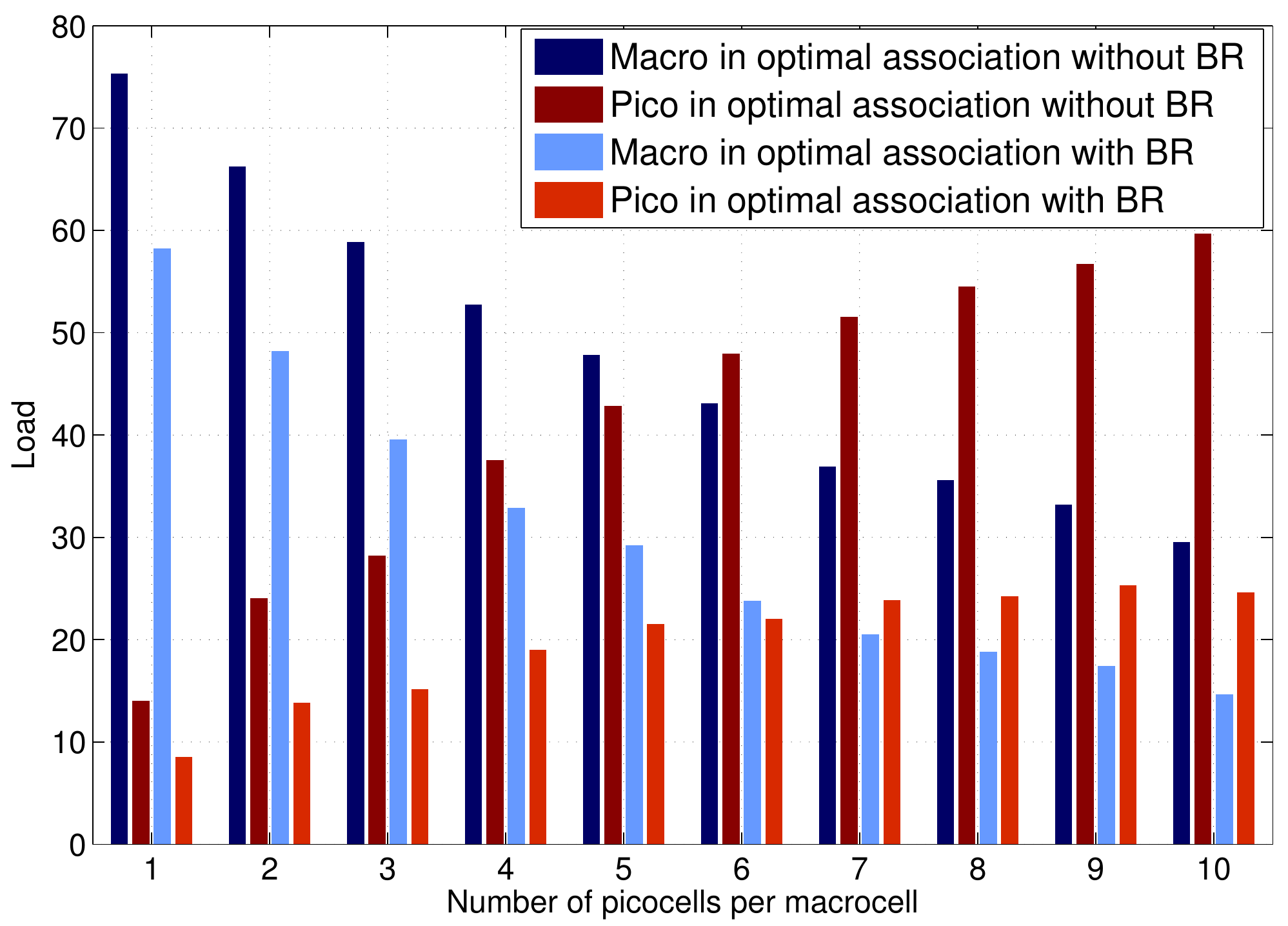}}
\caption{The relationship between load and small cell density in a two-tier HetNet.}
\label{fig:load}
\end{figure}

The performance of  a three-tier network using different association schemes is compared in Fig. \ref{fig:sinr}. We compare five different association approaches, among which the ``Max-SINR in normal RBs with BR'' is a scheme where the association is based on the signal received in normal RBs and the association in blank RBs is kept the same, even though some BSs are turned off. By jointly adopting BR, the load-aware association further improves the network performance (e.g., 5x gain for worst 5\% users compared to Max-SINR without BR). Fig. \ref{fig:sinr} also indicates the importance of appropriate association in networks with BR. By adopting BR with inappropriate association (e.g., Max-SINR association based on the received SINR in normal RBs), the network performance may even be degraded. On the other hand, with appropriate association (need not be optimal), the gain can be significant (e.g., 3x gain for worst 5\% users adopting Max-SINR with BR compared to Max-SINR without BR, and 5x gain compared to Max-SINR in normal RBs with BR).

\begin{figure}
\centering
\includegraphics[width=8.4cm, height=6.0cm]{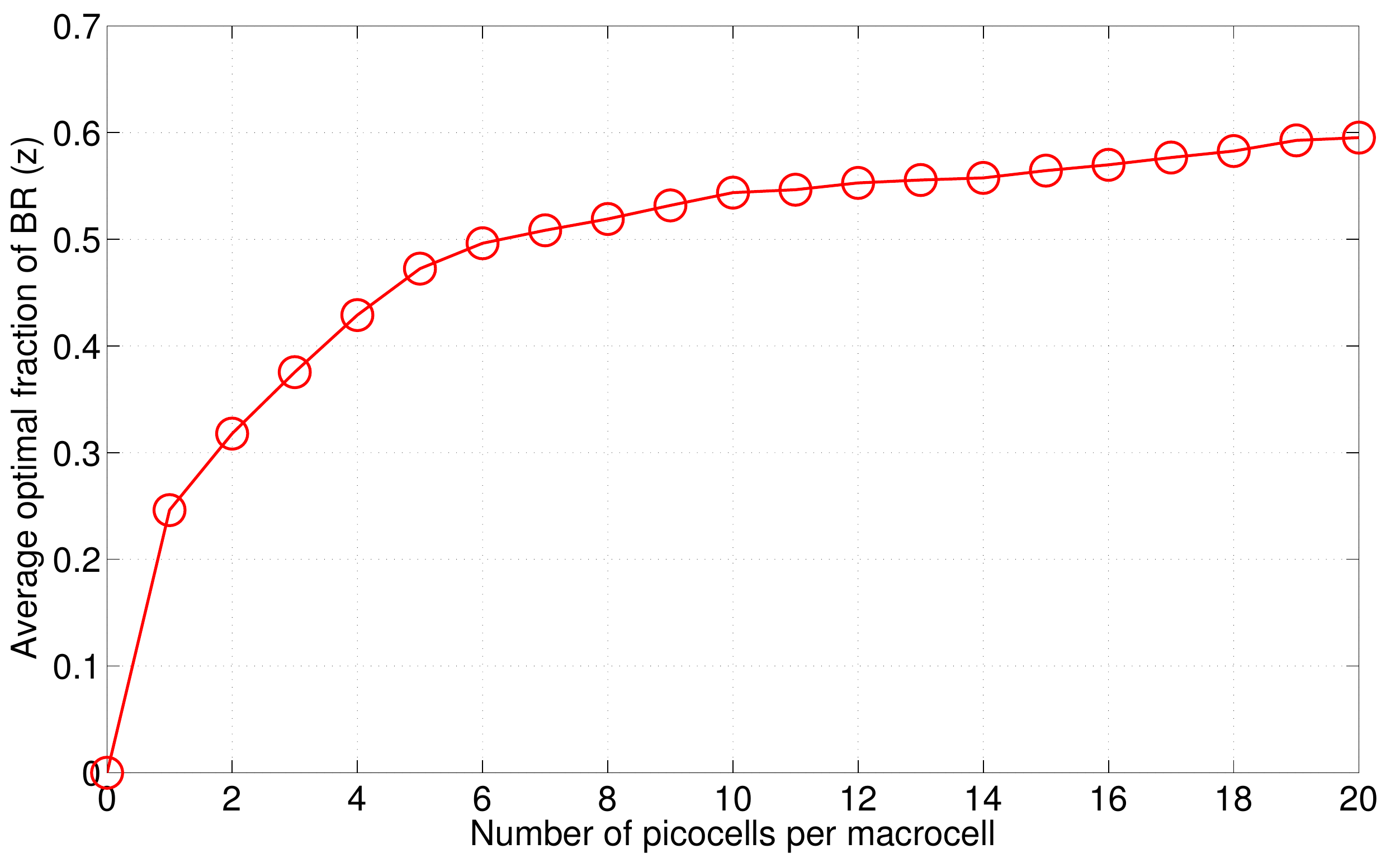}
\caption{The average of optimal fraction of blank RBs (i.e., $z$) vs. density of small cells.}
\label{fig:abs}
\end{figure}

To investigate the impact of different densities of small cells in HetNets, we consider a two-tier network consisting of macrocells and picocells. The fraction of BR in different network settings is shown in Fig. \ref{fig:abs}, where we average the optimal fraction of BR over different realizations. In Fig. \ref{fig:abs}, the fraction of BR increases when the small cells become denser. 

We compare the load in optimal resource allocation with BR to the optimal association without eICIC in Fig. \ref{fig:load}. With the increase of picocells, the load of macrocells keeps decreasing and more users are pushed off to small cells, as shown in Fig. \ref{fig:loadnormal}. Note that adopting the BR approach, a user can be served in blank and/or normal RBs. We have proven that the associations in blank and normal RBs are very different. From simulation, we observe that only a very small fraction of users are served both in blank and normal RBs, as illustrated in Fig. \ref{fig:abseg}. While more users are pushed off to small cells in both approaches (optimal resource allocation with BR and  without BR) when the density of picocells increases, Fig. \ref{fig:loadabs} shows that fewer users are served by small cells in blank RBs compared to the optimal association without eICIC. One possible reason is that as picocells become denser,  many users served by small cells in normal RBs already have a good enough rate, so the gain from turning off macro BSs decreases, which provide less motivation to push off users to picocells in blank RBs. The diminishing gain can also be observed in Fig.~\ref{fig:throughput} as the density of picocell increases.

In Fig. \ref{fig:throughput}, we show the throughput gain of cell-edge users in different network deployments. The gain is compared to the optimal association without eICIC (i.e., $\frac{T_{a}-T_{n}}{T_{n}}$, where $T_{a}$ and $T_{n}$ are the throughput of worst $10\%$ users using optimal resource allocation with BR and without BR, respectively). Different from the gain compared to Max-SINR, the gain here implies the potential impact of BR on the performance improvement. When the picocells become denser, it is more necessary to turn off the macrocells, but the gain from BR decreases. In a sparse network, the main interference is from macrocells,  and thus the potential gain in SINR by turning off macro BSs is large. When the network is increasingly dense, the aggregate interference from small cells keeps increasing and may even overwhelm the interference from macrocells. In this case, though there is still gain from BR,  the SINR improvement of users in small cells decreases due to the large interference from other small cells. Therefore, the gain depends on the aggregate interference from small cells, and thus depends mainly on the transmit power of small cells. Since femtocells have much lower power, the gain keeps increasing as the density of femtocells increases, which is quite different from picocells. 

\begin{figure}
\centering
\includegraphics[width=8.5cm, height=6.0cm]{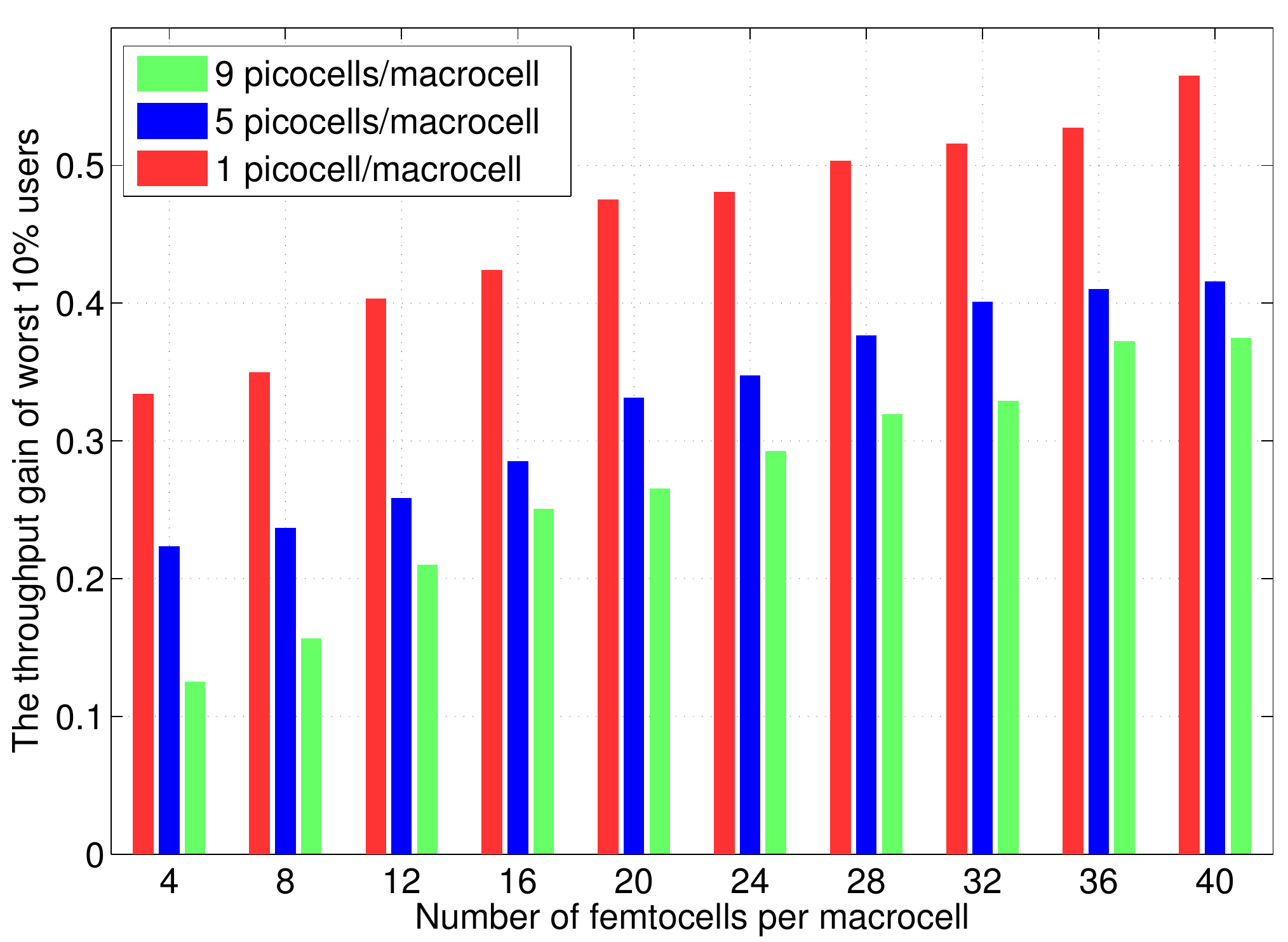}
\caption{The throughput gain of worst 10\% users in networks with different deployments of small cells.}
\label{fig:throughput}
\end{figure}

\section{Conclusion}\label{sec:conclusion}
In this paper, we propose a novel framework for joint optimization of BR and cell association in HetNets, which provides a large gain in network performance, in particular ``cell-edge'' rate. We formulated a network-wide utility maximization problem, which is converted to a convex optimization by single association relaxation. Although we allow users to be jointly served by multiple BSs, we proved that the number of fractional users is very small (at most $N_B-1$), and the simulations show that most users have single association. We also showed that association is  very different when adopting BR. Broadening the framework to more general settings (e.g., asynchronous configuration), analysis of the loss from fractional association to single association, investigating the gap between a simple biasing approach and the joint optimal solution, and efficient algorithms are left for future~work.


\bibliographystyle{ieeetr}
\bibliography{absbib}

\begin{thebibliography}{10}

\bibitem{GhoAnd12}
A.~Ghosh {\em et~al.}, ``Heterogeneous cellular networks: From theory to
  practice,'' {\em IEEE Communications Magazine}, vol.~50, pp.~54--64, June
  2012.

\bibitem{Bje11}
B.~Bjerke, ``{LTE}-advanced and the evolution of {LTE} deployments,'' {\em IEEE
  Wireless Communications}, vol.~18, pp.~4--5, Oct. 2011.

\bibitem{DamMon11}
A.~Damnjanovic, J.~Montojo, Y.~Wei, T.~Ji, T.~Luo, M.~Vajapeyam, T.~Yoo,
  O.~Song, and D.~Malladi, ``A survey on {3GPP} heterogeneous networks,'' {\em
  IEEE Wireless Communications Magazine}, vol.~18, pp.~10--21, June 2011.

\bibitem{LopGuv11}
D.~Lopez-Perez, I.~Guvenc, G.~De~La~Roche, M.~Kountouris, T.~Q. Quek, and
  J.~Zhang, ``Enhanced intercell interference coordination challenges in
  heterogeneous networks,'' {\em IEEE Wireless Communications}, vol.~18,
  pp.~22--30, June 2011.

\bibitem{YeRon12}
Q.~Ye, B.~Rong, Y.~Chen, M.~Al-Shalash, C.~Caramanis, and J.~G. Andrews, ``User
  association for load balancing in heterogeneous cellular networks,'' {\em to
  appear, IEEE Trans. on Wireless Communications}, 2013.

\bibitem{MadBor10}
R.~Madan, J.~Borran, A.~Sampath, N.~Bhushan, A.~Khandekar, and T.~Ji, ``Cell
  association and interference coordination in heterogeneous {LTE}-{A} cellular
  networks,'' {\em IEEE Journal on Selected Areas in Communications}, vol.~28,
  pp.~1479--1489, Dec. 2010.

\bibitem{JoSan11}
H.~S. Jo, Y.~Sang, P.~Xia, and J.~Andrews, ``Heterogeneous cellular networks
  with flexible cell association: a comprehensive downlink {SINR} analysis,''
  {\em IEEE Trans. on Wireless Communications}, vol.~11, pp.~3484--3495, Oct.
  2012.

\bibitem{GuvJeo11}
I.~Guvenc, M.-R. Jeong, I.~Demirdogen, B.~Kecicioglu, and F.~Watanabe, ``Range
  expansion and inter-cell interference coordination (icic) for picocell
  networks,'' in {\em Proc., IEEE Veh. Technology Conf.}, pp.~1--6, Sep. 2011.

\bibitem{OhHan12}
J.~Oh and Y.~Han, ``Cell selection for range expansion with almost blank
  subframe in heterogeneous networks,'' in {\em Proc., IEEE PIMRC},
  pp.~653--657, Sep. 2012.

\bibitem{StaWic09}
S.~Stanczak, M.~Wiczanowski, and H.~Boche, {\em {F}undamentals of {R}esource
  {A}llocation in {W}ireless {N}etworks: {T}heory and {A}lgorithms}, vol.~3.
\newblock Springer Verlag, 2009.

\bibitem{GajHua09}
V.~Gajic, J.~Huang, and B.~Rimoldi, ``Competition of wireless providers for
  atomic users: Equilibrium and social optimality,'' in {\em Proc., Allerton
  Conf. on Comm., Control, and Computing}, pp.~1203--1210, Sep. 2009.

\bibitem{BoyVan04}
S.~Boyd and L.~Vandenberghe, {\em {C}onvex {O}ptimization}.
\newblock Cambridge University Press, 2004.

\end{thebibliography}
\end{document}